\documentclass[11pt]{article}

\usepackage[margin=1in]{geometry}
\usepackage[english]{babel}
\usepackage[utf8x]{inputenc}
\usepackage{amsmath}
\usepackage{amsfonts}
\usepackage{soul}
\usepackage{amssymb}
\usepackage{graphicx}
\usepackage{rotating}
\usepackage{color}
\usepackage{amsbsy}
\usepackage{subfig}
\usepackage{graphicx,  ucs}
\usepackage{float}
\usepackage{tikz}
\usepackage{algorithm}
\usepackage[noend]{algpseudocode}
\usepackage{listings}
\usepackage{amsthm}
\usepackage{enumitem}
\usepackage{hyperref}
\usepackage{multirow}
\usepackage{array}
\usepackage[compact]{titlesec}
\usepackage{cmap}
\usepackage[T1]{fontenc}
\usepackage{bm}
\def\compactify{\itemsep=0pt \topsep=0pt \partopsep=0pt \parsep=0pt}
\let\latexusecounter=\usecounter
\newenvironment{CompactEnumerate}
  {\def\usecounter{\compactify\latexusecounter}
   \begin{enumerate}}
  {\end{enumerate}\let\usecounter=\latexusecounter}
{\begin{itemize}%
\setlength{\itemsep}{0pt}%
\setlength{\topsep}{0pt}%
\setlength{\partopsep}{0 in}%
\setlength{\parskip}{0 pt}}%
{\end{itemize}}

\newenvironment{pproof}
{ \noindent \textit{Proof.}  }
{ \hfill \rule{1.5ex}{1.5ex} }

\newcommand{\Domain}{\mathcal{D}}
\newcommand{\Support}{\mathcal{X}}

\newcommand{\tl}{\textlatin}

\newcommand{\com}[1]{}

\def\poly{\mathrm{poly}}
\def\eps{\varepsilon}

\def\floor#1{\mathop{\left\lfloor#1\right\rfloor}}
\def\ceil#1{\mathop{\left\lceil#1\right\rceil}}

\def\PLS{\textsc{PLS}}

\def\PPAD{\textsc{PPAD}}

\def\CLS{\textsc{CLS}}

\def\Banach{\textsc{MetricBanach}}
\def\Banachh{\textsc{Banach}}
\def\clocal{\textsc{Continuous LocalOpt}}
\def\BIP{Basic Iterative Procedure}
\def\Bsf#1{B\left(#1\right)}

\newcommand{\reals}{\mathbb{R}}
\newcommand{\nats}{\mathbb{N}}

\def\Int{\mathrm{Int}}
\def\Clos{\mathrm{Clos}}

\def\pot{p}

\def\norm#1{\left\|#1\right\|}

\def\abs#1{\left|#1\right|}

\newtheorem{theorem}{Theorem}
\newtheorem*{BanachFPT}{Banach's Fixed Point Theorem}
\newtheorem*{BessagaCT}{Bessaga's Converse Theorem}
\newtheorem*{MeyersCT}{Meyers's Converse Theorem}

\newtheorem{proposition}{Proposition}

\newtheorem{definition}{Definition}
\newtheorem{corollary}{Corollary}
\newtheorem{lemma}{Lemma}
\newtheorem{remark}{Remark}

\renewcommand{\vec}{\bm}

\newcommand{\diam}[2]{\mathrm{diam}_{#1}\left[ #2 \right]}

\definecolor{vergreen}{RGB}{0,85,2}
\definecolor{myvergreen}{RGB}{0,140,3}
\definecolor{provorange}{RGB}{85,34,0}
\definecolor{inputblue}{RGB}{5,13,111}
\definecolor{noapred}{RGB}{116,3,3}
\definecolor{classesblue}{RGB}{9,49,146}

\definecolor{secinhead}{RGB}{249,196,95}

\definecolor{lgray}{gray}{0.8}

\newenvironment{prevproof}[2]{\noindent {\em {Proof of {#1}~\ref{#2}:}}}{$\hfill\qed$\vskip \belowdisplayskip}

\def\compactify{\itemsep=0pt \topsep=0pt \partopsep=0pt \parsep=0pt}
\let\latexusecounter=\usecounter

\newenvironment{Enumerate}
  {\def\usecounter{\compactify\latexusecounter}
   \begin{enumerate}}
  {\end{enumerate}\let\usecounter=\latexusecounter}

\begin{document}

\sethlcolor{lgray}

%
%

\title{A Converse to Banach's Fixed Point Theorem\\ and its CLS Completeness} 
\date{}
\author{
  Constantinos Daskalakis \\ EECS and CSAIL, MIT \\ \href{mailto:costis@mit.edu}{costis@mit.edu}
  \and Christos Tzamos \\ EECS and CSAIL, MIT \\ \href{mailto:tzamos@mit.edu}{tzamos@mit.edu}
	\and Manolis Zampetakis \\ EECS and CSAIL, MIT \\ \href{mailto:mzampet@mit.edu}{mzampet@mit.edu}
	}
\clearpage
\maketitle

\begin{abstract}

    Banach's fixed point theorem for contraction maps has been widely used to analyze the convergence of iterative methods in non-convex problems. It is a common experience, however, that iterative maps fail to be globally contracting under the natural metric in their domain, making the applicability of Banach's theorem limited. We explore how generally we can apply Banach's fixed point
  theorem to establish the convergence of iterative methods when pairing it with carefully designed metrics. 
	
	Our first result is a strong converse of Banach's theorem,
  showing that it is a {\em universal analysis tool} for establishing global convergence of iterative methods to unique fixed points, and for bounding their convergence rate. In other words, we show that, whenever an iterative map globally converges to a unique fixed point, there exists a metric under which the iterative map is contracting and which can be used to bound the number of iterations until convergence. We
  illustrate our approach in the widely used power method, providing a new way of bounding its convergence rate through contraction arguments.

  We next consider the computational complexity of Banach's fixed point theorem. Making the proof of our converse theorem constructive, we show that computing a fixed point whose existence is guaranteed by Banach's fixed point theorem is CLS-complete. We thus provide the first natural complete problem for the class CLS, which was defined in~\cite{DP11} to capture the complexity of problems such as P-matrix LCP, computing KKT-points, and finding mixed Nash equilibria in congestion and network coordination games.
\end{abstract}
\thispagestyle{empty}

\addtocounter{page}{-1}\newpage
\section{Introduction}
\label{sec:intro}

Several widely used computational methods are fixed point iteration methods. These include gradient descent, the power iteration method, alternating optimization, the expectation-maximization algorithm, $k$-means clustering, and others. In several important applications, we have theoretical guarantees for the convergence of these methods. For example, convergence to a unique solution can be guaranteed when the method is explicitly, or can be related to, gradient descent on a convex function~\cite{nemirovskinotes,nesterovnotes,boydnotes}. More broadly, convergence to a stationary point can be guaranteed when the method is, or can be related to, gradient descent; for some interesting recent work on the limit points of gradient descent, see~\cite{Piliouras,Jordan} and their references. 

Another, more general, style of analysis for proving convergence of fixed point iteration methods is via a potential (a.k.a. Lyapunov) function. For example, analyzing the power iteration method amounts to showing that, as time progresses, the unit vector maintained by the algorithm places more and more of its $\ell_2$ energy on the principle eigenvector of the matrix used in the iteration, if it is unique, or, anyways, on the eigenspace spanned by the principal eigenvectors; see Appendix~\ref{sec:app:power} for a refresher. In passing, it should also be noted that the power iteration method itself is commonly used as a tool for establishing the convergence of other fixed point iteration methods, such as alternating optimization; e.g.~\cite{Moritz}.

Ultimately, all fixed point iteration methods aim at converging to a fixed point of their iteration map. For global convergence to a unique solution, it should also be the case that the fixed point of the iteration map is unique. It is, thus, unsurprising that another widely used approach for establishing convergence of these methods is by appealing to Banach's fixed point theorem. To recall, consider an iteration map $x_{t+1} \leftarrow f(x_t)$, where $f: {\cal D} \rightarrow {\cal D}$, and suppose that there is a distance metric $d$ such that $({\cal D},d)$ is a complete metric space and $f$ is contracting with respect to $d$, i.e.~for some constant $c<1$, $d(f(x),f(y)) \le c\cdot d(x,y)$, for all $x,y \in {\cal D}$. Under this condition, Banach's fixed point theorem guarantees that there is a {\em unique} fixed point $x^* = f(x^*)$. Moreover, iterating $f$ is bound to converge to $x^*$. Specifically, the $t$-fold composition, $f^{[t]}$, of $f$ with itself satisfies: $d(f^{[t]}(x_0),x^*) \le c^t d(x_0,x^*)$, for any starting point $x_0$. 

Given Banach's theorem, if you established that your iteration method is contracting under some distance metric $d$, you would also have immediately proven that your method converges and that it may only converge to a unique point. Moreover, you can predict how many steps you need from any starting point $x_0$ to reach an approximate fixed point $x$ satisfying $d(f(x),x) <\epsilon$ for some accuracy~$\epsilon$.\footnote{Indeed, it can be easily shown that $d(f^{[t+1]}(x_0),f^{[t]}(x_0))\le c^t d(x_1,x_0)$. So $t = \log_{1/c} {d(x_1,x_0) \over \epsilon}$ steps suffice.} Alas, several widely used fixed point iteration methods are not generally contracting, or only contracting in a small neighborhood around their fixed points and not the entire domain where they are defined. At least, this is typically the case for the metric $d$ under which approximate fixed points, $d(f(x),x) <\epsilon$, are sought. There is also quite an important reason why they may not be contracting: several of these methods may in fact have multiple fixed points.

Given the above motivation, our goal in this paper is to {\em investigate the extent to which Banach's fixed point theorem is a {\em universal analysis tool} for establishing that a fixed point iteration method both converges and globally converges to a unique fixed point.} More precisely, our question is the following: if an iterative map $x_{t+1} \leftarrow f(x_t)$ for some $f: {\cal D} \rightarrow {\cal D}$ converges to a unique fixed point $x^*$ from any starting point, is there always a way to prove this using Banach's fixed point theorem? Additionally, can we always use Banach's fixed point theorem to compute how many iterations we would need to find an approximate fixed point $x$ of $f$ satisfying $d(x,f(x))<\epsilon$, for some distance metric $d$ and accuracy $\epsilon>0$?

We study this question from both a mathematical and a computational perspective. On the mathematical side, we show a strong converse of Banach's fixed point theorem, saying the following: given an iterative map $x_{t+1} \leftarrow f(x_t)$ for some $f: {\cal D} \rightarrow {\cal D}$, some distance metric $d$ on ${\cal D}$ such that $({\cal D},d)$ is a complete and proper metric space, and some accuracy $\epsilon >0$, if $f$ has a unique fixed point that the $f$-iteration converges to from any starting point, then for any constant $c \in (0,1)$, there exists a distance metric $d_c$ on $\cal D$ such that:
\begin{enumerate}
\item $d_c$ certifies uniqueness and convergence to the fixed point, by satisfying $d_c(f(x),f(y)) \le c \cdot d_c(x,y)$, for all $x, y \in {\cal D}$; \label{property 1}
\item $d_c$ allows an analyst to predict how many iterations of $f$ would suffice to arrive at an approximate fixed point $x$ satisfying $d(x,f(x))<\epsilon$; notice in particular that we are interested in finding an approximate fixed point with respect to the original distance metric $d$ (and not the constructed one $d_c$). \label{property 2}
\end{enumerate}
Our converse theorem is formally stated as Theorem~\ref{th:converse} in Section~\ref{sec:BanachConverse}. In the same section we discuss its relationship to other known converses of Banach's theorem known in the literature, in particular Bessaga's and Meyers's converse theorems. The improvement over these converses is that our constructed metric $d_c$ is such that it allows us to bound the number of steps requied to reach an approximate fixed point according to the metric of interest $d$ and not just $d_c$; namely Property~\ref{property 2} above holds. We discuss this further in Section~\ref{sec:corollaries of new converse}. Section~\ref{sec:new converse} provides a sketch of the proof, and the complete details can be found in Appendix~\ref{sec:app:mainproof}.

While the proof of Theorem~\ref{th:converse} is non-constructive, it does imply that Banach's fixed point theorem is a universal analysis tool for establishing global convergence of fixed point iteration methods to unique solutions. Namely, it implies that one can {\em always} find a witnessing metric. We illustrate this by studying an important such method: power iteration. The power iteration method is a widely-used and well-understood method for computing the eigenvalues and eigenvectors of a matrix. It is well known that if a matrix $A$ has a unique principal eigenvector, then the power method starting from a vector non-perpendicular to the principal eigenvector will converge to it. This is shown using a potential function argument, outlined above and in Appendix~\ref{sec:app:power}, which also pins down the rate of convergence. 

Our converse to Banach's theorem, guarantees that, besides the potential function argument, there must also exist a distance metric under which the power iteration is a contraction map. Such a distance metric is not obvious, as contraction under any $\ell_p$-norm fails; we provide counter-examples in Section~\ref{sec:powerMethod}. To illustrate our theorem, we identify a new distance metric under which the power method is indeed contracting at the optimal rate. See Proposition~\ref{prop:powerContr}. 
Our distance metric serves as an alternative proof for establishing that the power iteration converges and for pinning down its convergence rate. 

We close the circle by studying Banach's fixed point theorem from a computational standpoint. Recent work of Daskalakis and Papadimitriou~\cite{DP11} has identified a complexity class, CLS, where finding a Banach fixed point lies. CLS, defined formally in Section~\ref{sec:cls}, is a complexity class at the intersection of PLS~\cite{JPY88} and PPAD~\cite{Papadimitriou94}. Roughly speaking, PLS contains total problems whose existence of solutions is guaranteed by a potential function argument, while PPAD contains total problems whose existence of solutions is guaranteed by Brouwer's fixed point theorem. Lots of interesting  work has been done on both classes in the past two decades; for a small sample see e.g.
~\cite{DGP09,ChenDT09,Aviad16,FPT06,Skopalik,priortoangel,AngelBPW17} 
and their references. CLS, lying in the intersection of PLS and PPAD, contains comptutational problems whose existence of solutions is guaranteed by both a potential function and a fixed point argument.\footnote{More precisely, it contains all problems reducible to~\clocal, defined in Section~\ref{sec:cls}, and which doesn't necessarily capture the whole intersection of PPAD and PLS.}
 
Unsurprisingly CLS contains several interesting problems, whose complexity is not known to lie in P, but which also are unlikely to be complete for PPAD or PLS. One of these problems is finding a Banach fixed point. Others include the P-matrix Linear Complementarity Problem, finding mixed Nash equilibria of network coordination and congestion games, computational problems related to finding KKT points, and solving Simple Stochastic Games; see~\cite{DP11} for precise definitions of these problems and for references. Moreover, recent work has provided cryptographic hardness results for CLS~\cite{CLScrypto} based on obfuscation, extending work which proved cryptographic hardness results for PPAD~\cite{Nashcrypto,Rosen,moni}. 

Ultimately, the definition of CLS was inspired by a vast range of total problems that could not be properly classified as complete in PPAD or PLS due to the nature of their totality arguments. However, no natural complete problem for this class has been identified, besides \clocal, through which the class was defined. By making our converse to Banach's fixed point theorem constructive, we show that finding a Banach fixed point is CLS-complete. More precisely, in Section~\ref{sec:cls} we define problem $\Banachh$, whose input is a continuous function $f$ and a continuous metric $d$, and whose goal is to either output an approximate fixed point of $f$ 
or a violation of the contraction of $f$ with respect to $d$. 
In Theorem~\ref{th:cBanach2} we show that $\Banachh$ is CLS-complete.\footnote{It is worth pointing out that, while some problems in CLS (e.g. Banach fixed points, simple stochastic games) have unique solutions, most do not. Given that contraction maps have unique fixed points, the way we bypass the potential oxymoron, is by accepting as solutions violations of contraction.}


\paragraph{Further Related Work.} We note that contemporaneously and independently from our work, Fearnley et al.~\cite{FearnleyGMS17} have also identified a CLS-complete problem related to Banach's fixed point theorem. Their problem, called {\sc MetametricContraction}, takes as input a function $f$ and a metametric $d$, and asks to find an approximate fixed point of $f$, 
or a violation of the contraction of $f$ with respect to $d$. In comparison to our CLS-completeness results, the CLS-hardness of {\sc Banach} in our paper is stronger than that of {\sc MetametricContraction} as the input to {\sc Banach} is a metric. On the other hand, the containment of {\sc MetametricContraction} into CLS is stronger than the containment of {\sc Banach}, as 
{\sc Banach} is polynomial-time reducible to {\sc MetametricContraction}.
\section{Notation and Preliminaries}
\label{sec:model}

\noindent \textbf{Basic Notation} We use $\reals_+$ to refer to set of non-negative real numbers and $\nats_1$ is the set of natural numbers except $0$. We call a
function $f$ \textit{selfmap} if it maps a domain $\Domain$ to itself, i.e. $f : \Domain \to \Domain$. For a selfmap $f$ we use $f^{[n]}$ to refer to the $n$ times
composition $f$ with it self, i.e. $\underbrace{f(f(\dots f(\cdot)))}_{n \text{ times }}$.

\medskip
  We use $\norm{\cdot}_p$ to refer to the $\ell_p$ norm of a vector in $\reals^n$. We use $\Domain/\sim$ to refer to the set of equivalence classes of the
equivalence relation $\sim$ on a set $\Domain$. Finally, we use $S^*$ to refer to the Kleene star of a set $S$.

\medskip
  A real valued function $g : \Domain^2 \to \reals$ is called \textit{symmetric} if $g(x, y) = g(y, x)$ and \textit{anti-symmetric} if $g(x, y) = - g(y, x)$.

\medskip
  In Appendix \ref{sec:app:def}, the reader can find some well-known definitions that we are going to use in the rest of the paper. More precisely in the field
of:

\medskip
\textbf{Topological Spaces}, we define the notion of: topology, topological spaces, open sets, closed sets, interior of a set $A$, denoted $\Int(A)$, closure of
a set $A$, denoted $\Clos(A)$.

\medskip
\textbf{Metric Spaces}, we define the notion of: distance metric, metric space, diameter, bounded metric space, continuous function, open and closed sets in a metric space, compact set, locally compact metric space, proper metric space, open and closed balls, Cauchy sequence, complete metric space, equivalent metrics, continuity, Lipschitz continuity, contraction property, fixed point.

\smallskip Because of its importance for the rest of the paper we also give here the definition of a distance metric and metric space.
\begin{definition}
  \em
  \label{def:dMetric}
  Let $\Domain$ be a set and $d : \Domain^2 \to \reals$ a function with the following properties:
  \begin{CompactEnumerate}[label=(\roman*)]
    \label{eq:dMetric}
    \item $d(x, y) \ge 0$ for all $x, y \in \Domain$. \label{eq:dMetric1}
    \item $d(x, y) = 0$ if and only if $x = y$. \label{eq:dMetric2}
    \item $d(x, y) = d(y, x)$ for all $x, y \in \Domain$. \label{eq:dMetric3}
    \item $d(x, y) \le d(x, z) + d(z, x)$ for all $x, y, z \in \Domain$. This is called \textit{triangle inequality}. \label{eq:dMetric4}
  \end{CompactEnumerate}
  Then we say that $d$ is a {\em metric} on $\Domain$, and $(\Domain, d)$ is a {\em metric space}.
\end{definition}

\smallskip
\paragraph{\BIP.} If a selfmap $f$ has a fixed point and is continuous, we can define the following sequence of points $x_{n + 1} = f(x_n)$ where
the starting point $x_0$ can be picked arbitrarily. If $(x_n)$ converges to a point $\bar{x}$ then
\[ \lim_{n \rightarrow \infty} x_{n + 1} = \lim_{n \rightarrow \infty} f(x_n) \Rightarrow \lim_{n \rightarrow \infty} x_{n + 1} = f\left(\lim_{n \rightarrow \infty} x_n\right) \Rightarrow \bar{x} = f(\bar{x}). \]

\noindent This observation implies that a candidate procedure for computing a fixed point of a selfmap $f$ is to \textit{iteratively} apply the function
$f$ starting from an arbitrary point $x_0$. If this procedure converges then the limit is a fixed point $x^*$ of $f$. We will refer to this method of computing fixed
points as the \textit{\BIP}.

\medskip
\paragraph{Arithmetic Circuits.} In Section \ref{sec:cls} we work with functions from continuous domains to continuous domains represented as
\textit{arithmetic circuits}. An arithmetic circuit is defined by a directed acyclic graph (DAG). The inputs to the circuit are in-degree $0$ nodes, and the outputs are out-degree $0$ nodes. Each non-input node is a gate from the set $\{+,-,*,\max,\min,>\}$, performing an operation on the outputs of its in-neighbors. The meaning of the ``$>$'' gate is $>(x, y) = 1$ if $x > y$ and $0$ otherwise. We also allow ``output a rational constant'' gates. These are gates without any inputs, which output a rational constant.

\section{Converse Banach Fixed Point Theorems}
\label{sec:BanachConverse}
	
	We start, in Section~\ref{sec:known converses}, with an overview of known converses to Banach's fixed point theorem. We also explain why these converses are not enough to prove that
Banach's fixed point theorem is a universal tool for analyzing the convergence of iterative algorithms. Then, in Section~\ref{sec:new converse}, we prove a stronger converse theorem that demonstrates
the universality of Banach's fixed point theorem for the analysis of iterative algorithms. Before beginning, we formally state  Banach's fixed Point Theorem. A  useful survey of the applications of this theorem can be found in \cite{conrad14}.



\begin{BanachFPT}
\label{th:Banach1}
Suppose $d$ is a distance metric function such that $(\Domain, d)$ is a complete metric space, and suppose that $f: \Domain \rightarrow \Domain$ is a contraction map according to $d$, i.e.
\begin{equation}
\label{eq:contraction}
 d(f(x), f(y)) \le c \cdot d(x, y), \forall x,y, \emph{ for some } c < 1.
\end{equation}

\noindent Then $f$ has a unique fixed point $x^*$ and the convergence rate of the \BIP~with respect to $d$ is $c$. That is, $d(f^{[n]}(x_0),x^*)<c^n \cdot d(x_0,x^*)$, for all $x_0$.
\end{BanachFPT}

\subsection{Known Converses to Banach's Fixed Point Theorem} \label{sec:known converses}

The first known converse to Banach's fixed point theorem is the following \cite{Bessaga1959}.

\begin{BessagaCT}
    Let $f$ be a map from $\Domain$ to itself, and suppose that $f^{[n]}$ has unique fixed point for every $n \in \nats_1$. Then, for every constant $c \in (0, 1)$, there exists
  a distance metric $d_c$ such that $(\Domain, d_c)$ is a complete metric space and $f$ is a contraction map with respect to $d_c$ with contraction constant $c$.
\end{BessagaCT}

  The implication of the above theorem is that, if we want to prove existence and uniqueness of fixed points of $f^{[n]}$ for all $n$, then Banach's fixed point theorem is a
universal way to do it. Moreover, there is a potential function of the form $\pot(x) = d_c(x, f(x))$, where $d_c$ is a distance metric, that decreases under successive applications of $f$, and successive applications of $f$ starting from any point $x_0$ are bound to converge to the unique fixed point of $f$.

Unfortunately, $d_c$ cannot provide any information about the number of steps that the \BIP~needs before computing an approximate fixed point under some metric $d$ of interest. The reason is that, after $\log_c \eps$ steps of the \BIP, we only have $d_c(x_n, f(x_n)) \le \eps$. However, $d_c$ might not have any relation to $d$, hence an approximate fixed point under $d_c$ may not be one for $d$. So Bessaga's theorem is not useful for bounding the running time of iterative methods for approximate fixed point computation.

  \smallskip Given the above discussion, it is reasonable to expect that a converse to Banach's theorem that is useful for bounding the running time of approximate fixed point computation methods should take into account, besides the function $f$ and its domain $\Domain$, the distance metric $d$ under which we are interested in computing approximate fixed points. One step in this direction has already been made by Meyers \cite{Meyers1967}.

\begin{MeyersCT}
  Let $(\Domain, d)$ be a complete metric space, where $\Domain$ is compact, and suppose that $f: \Domain \rightarrow \Domain$ is continuous with respect to $d$. Suppose further that $f$ has a unique fixed point $x^*$, that the
Basic Iterative Method converges to $x^*$ from any starting point, and that there exists an open neighborhood $U$ of $x^*$ such that $f^{[n]}(U) \rightarrow \{x^*\}$. Then, for any $c \in (0, 1)$, there exists a distance metric $d_c$, which is topologically equivalent to $d$, such that
$(\Domain, d_c)$ is a complete metric space and $f$ is a contraction map with respect to $d_c$ with contraction $c$.
\end{MeyersCT}

  Compared to Bessaga's theorem, the improvement offered by Meyer's Theorem is that, instead of the existence of an arbitrary metric, it proves the existence of a metric that is topologically equivalent to the metric $d$. However, this is
still not enough to bound the number of steps needed by the \BIP~in order to arrive at a point $x_n$ such that $d(x_n, f(x_n)) \le \eps$. Our goal in the next section is to close this gap. We will also replace the compactness assumption with the assumption that $(\Domain, d)$ is proper, so that the converse holds for unbounded spaces.

\subsection{A New Converse to Banach's Fixed Point Theorem} \label{sec:new converse}

  The main technical idea behind our converse to Banach's fixed point theorem is to adapt the proof of Meyers's theorem to get a distance metric $d_c$ with the
property $d_c(x, y) \ge d(x, y)$ everywhere, except maybe for the region $d(x, x^*) \le \eps$. This implies that, if we guarantee that $d_c(x_n, x^*) \le \eps$, then
$d(x_n, x^*) \le \eps$.

\begin{theorem} \label{th:converse}
  Suppose $(\Domain, d)$ is a complete, proper metric space, $f : \Domain \to \Domain$ is continuous with respect to $d$ and the following hold:
\begin{CompactEnumerate}
\item $f$ has a unique fixed point $x^*$;
\item for every $x \in \Domain$, the sequence $(f^{[n]}(x))$ converges to $x^*$ with respect to $d$; moreover there exists an open neighborhood $U$ of $x^*$ such that $f^{[n]}(U) \to \{x^*\}$.
\end{CompactEnumerate}
Then, for every $c \in (0, 1)$ and  $\eps > 0$, there exists a distance metric function $d_{c, \eps}$ that is topologically equivalent to $d$ and is such that
$(\Domain, d_{c, \eps})$ is a complete metric space and
\begin{subequations}
\label{eq:converseCond}
  \begin{align}
    & \forall x,y \in \Domain: d_{c, \eps}(f(x), f(y)) \le c \cdot d_{c, \eps}(x, y); \label{eq:converseCond1} \\
    & \forall x,y \in \Domain: d_{c, \eps}(x, y) \le \eps \implies \min\{d(x^*, x), d(x^*, y), d(x, y)\} \le 2 \eps. \label{eq:converseCond2} 
  \end{align}
\end{subequations}
\end{theorem}

\paragraph{Remark.} Notice that the continuity of $f$ is a necessary assumption for the above statement to 
hold, as \eqref{eq:converseCond1} implies continuity given that $d_{c, \eps}$ and $d$ are topologically 
equivalent. Also the condition 2. of the theorem is implied by the existence of $d_{c, \eps}$ and it is not
true even if $f^{[n]}(x) \to x^*$ for any $n \in \nats$, since counter examples exist. Therefore this 
assumption is also necessary for our theorem to hold.

The proof of our Theorem~\ref{th:converse} adapts the construction of Meyers's proof, to ensure that \eqref{eq:converseCond2} 
is satisfied. We give here a proof sketch postponing the complete details to Appendix \ref{sec:app:mainproof},
where we repeat also all the technical details proven by Meyers \cite{Meyers1967}. 

\begin{proof}[Proof Sketch.]
  The construction of the metric $d_c$ follows is done in three steps:
\begin{CompactEnumerate}[label=\Roman*.]
  \item Starting from the original metric $d$, a non-expanding closure of $d$ is defined as the metric $d_M(x,y) = \sup_{i \ge 0} d( f^{(i)}(x), f^{(i)}(y) )$. This is topologically equivalent to $d$, but ensures that the images of any two points are at least as close in $d_M$ as the original two points (non-expanding property). 
  
  Notice that as $d_M(x, y) \ge d(x, y)$ for all points $x, y \in \Domain$, if we ensure that Property \eqref{eq:converseCond2} holds with respect to $d_M$ for  the final constructed metric $d_{c,\eps}$, it will also hold with respect to the original metric $d$.
        
  \item Given $d_M$, the construction proceeds by defining a function $\rho_{c, \eps}$ which satisfies~\eqref{eq:converseCond1}. This function achieves contraction by a constant $c<1$ by counting the number of steps required to reach an $\eps$-ball close to the fixed point.
  
  While for the original proof of Meyer any such $\eps$-ball suffices, in order to guarantee Property \eqref{eq:converseCond2}, our proof requires a set $S$ of points with small diameter with respect to $d$ such that performing an iteration of $f$ on any one of them results in a point still in the set $S$. We show that such a set always exists in Lemma \ref{lem:neighborhood} in Appendix~\ref{sec:app:mainproof}.
  
  This guarantees that $\rho_{c, \eps}(x, y) \ge d_M(x, y)$ if $\max\{d(x^*, x), d(x^*, y)\} \ge \eps$, and therefore Property \eqref{eq:converseCond2} is
        preserved.
        
        The function $\rho_{c, \eps}$ satisfies all required properties other than triangle inequality and thus is not a metric. However, it can be converted into one.
  \item Given $\rho_{c, \eps}$, we construct the sought after metric $d_{c, \eps}$ by taking it equal to the $\rho_{c, \eps}$-geodesic distance (metric closure of $\rho_{c, \eps}$). This directly converts $\rho_{c, \eps}$ into a metric. We show that after this operation Properties~\eqref{eq:converseCond1} and~\eqref{eq:converseCond2}. This is done in Lemma~\ref{lem:proofOursLemmaIII1} and Lemma~\ref{lem:proofOursLemmaIII2} in Appendix~\ref{sec:app:mainproof}.
\end{CompactEnumerate}
\end{proof}

\subsection{Corollaries of Theorem \ref{th:converse}} \label{sec:corollaries of new converse}

 Property \eqref{eq:converseCond2} of the metric output by Theorem \ref{th:converse} has some interesting corollaries that we would not be able to get using the known
converses to Banach's theorem discussed in Section~\ref{sec:known converses}. The first one is that we can now compute, from $d_{c,\eps}$, the number of iterations needed in order to get to within $\eps$ of the fixed
point $x^*$ of $f$ from any starting point $x_0 \in \Domain$.

\begin{corollary} \label{cor:converse1}
    Under the assumptions of Theorem~\ref{th:converse}, starting from a point $x_0 \in \Domain$, and for any constant $c \in (0, 1)$, the \BIP~finds a point $x$ such that $d(x, x^*) \le \eps$ after
    \[ \frac{\log(d_{c, {\eps/2}}(x_0, f(x_0))) + \log((2 - 2 c)/\eps)}{\log(1/c)} \]
   iterations, where $d_{c, \eps/2}$ is the metric guaranteed by Theorem~\ref{th:converse}.
\end{corollary}

  In Corollary \ref{cor:converse1}, for any given $\eps$ of interest, we have to identify a different  distance metric $d_{c, \eps/2}$, guaranteed by Theorem~\ref{th:converse}, to bound the number of steps required by the \BIP~to get to within $\eps$ from the fixed point. Sometimes we are interested in the explicit tradeoff between the number of steps required to get to the proximity of the fixed point and the amount of proximity $\eps$. To find such a tradeoff we
have to make additional assumptions on $f$. A mild assumption that is commonly satisfied by iterative procedures for non-convex problems is that the \BIP~\textit{locally converges} to the fixed point $x^*$. That is, if $x_0$ is appropriately close to $x^*$, then the \BIP~converges. A common way of proving local convergence is to prove that $f$ is
a contraction with respect to $d$ \textit{locally} for $x, y \in \bar{B}(x^*, \eps)$. Theorem \ref{th:converse} provides a way to extend this local contraction
property to the whole domain $\Domain$ and get an an explicit closed form of the tradeoff between the number of steps and $\eps$, as implied by the following result.

\begin{corollary}
  \label{cor:converse2}
    Under the assumptions of Theorem \ref{th:converse}, and the assumption that there exists $0 < c < 1$, $\delta > 0$ such that
  \[ d(f(x), f(y)) \le c \cdot d(x, y) \emph{ for all } x, y \in \bar{B}(x^*, \delta), \]
   starting from any point $x_0 \in \Domain$, the \BIP~finds a point $x$ such that $d(x, x^*) \le \eps$ after
  \[ \frac{\log(d_{c, \delta/2}(x_0, f(x_0))) + \log(1/\eps) + \log(1 - c) + 1}{\log(1/c)} + 1 \]
  iterations, where $d_{c, \delta/2}$ is the metric guaranteed by Theorem~\ref{th:converse}.
\end{corollary}

\section{Example: The Power Iteration as a Contraction Map}
\label{sec:powerMethod}

The results of the previous section imply that Banach's fixed point theorem is a universal analysis tool for establishing global convergence of fixed point iteration methods to unique solutions.
While the proof of Theorem~\ref{th:converse} is non-constructive, it does imply that one can {\em always} find a witnessing metric under which the iterative map is contracting.

In this section, we illustrate this possibility by studying an important iterative method, the power iteration. The power iteration method is a widely-used and well-understood method for computing the eigenvalues and eigenvectors of a matrix. For a given matrix $A$, it is defined as:
$$\vec x_{t+1} = \frac {A \vec x_{t}}{\norm{A \vec x_{t}}_2}$$

It is well known that if a matrix $A$ has a unique principal eigenvector, then the power method starting from a vector non-perpendicular to the principal eigenvector will converge to it. This is shown using a potential function argument, presented in Appendix~\ref{sec:app:power}, which also pins down the rate of convergence.

Our converse to Banach's theorem, guarantees that, besides the potential function argument, there must also exist a distance metric under which the power iteration is a contraction map.
To illustrate our theorem, we identify a new distance metric under which the power method is indeed contracting at the optimal rate.

Such a distance metric is not obvious. As the following counterexample shows, contraction under any $\ell_p$-norm fails.

\paragraph{Counterexamples for $\norm{\cdot}_p$.} We show a counter example for $\ell_2$ norm which directly extends to any
$\ell_p$ norm. In particular, let $n = 2$, $\lambda_1 = 2$, $\lambda_2 = 1$ and the corresponding eigenvectors be $e_1=(1,0)$ and $e_2=(0,1)$.
The power iteration is given by $f(\vec x) = \frac{(2x_1, x_2)}{\sqrt{4x_1^2 + x_2^2}}$. We set 
$\vec x = \left( \frac {1} {\sqrt{5}} , \frac {2} {\sqrt{5}} \right)$
and
$\vec y = \left( \frac {1} {\sqrt{10}} , \frac {3} {\sqrt{10}} \right)$. 
We get that $\norm{f(\vec x) - f(\vec y)}_2 = \norm{ \left( \frac {1} {\sqrt{2}} , \frac {1} {\sqrt{2}} \right) - \left( \frac {2} {\sqrt{13}} , \frac {3} {\sqrt{13}} \right)   }_2 \ge 0.19$. Also
$\norm{\vec x - \vec y}_2 = \norm{ \left( \frac {1} {\sqrt{5}} , \frac {2} {\sqrt{5}} \right) - \left( \frac {1} {\sqrt{10}} , \frac {3} {\sqrt{10}} \right) }_2 \le 0.14$ and therefore $\norm{f(\vec x) - f(\vec y)}_2 > \norm{\vec x - \vec y}_2$.

Even though contraction is not achieved under any $\ell_p$-norm, it is possible to construct a metric under which power iteration is contracting even at the optimal rate which is given by the ratio of the two largest eigenvalues of matrix $A$. Our next theorem constructs such a metric.

\begin{proposition}
\label{prop:powerContr}
Let $A \in \mathbb{R}^{n\times n}$ be a matrix with left eigenvector-eigenvalue pairs $(\lambda_1,\vec v_1),...,(\lambda_n,\vec v_n)$ such that
$\lambda_1 > \lambda_2 \ge ... \ge \lambda_n$. Then the power iteration,
$\vec x_{t+1} = f(\vec x_t) \triangleq \frac {A \vec x_{t}}{\norm{A \vec x_{t}}}$
is contracting under the metric
$d(\vec x,\vec y) = \norm{\frac {\vec x} {\langle \vec x, \vec v_1 \rangle} - \frac {\vec y} {\langle \vec y, \vec v_1 \rangle}}_2$
with contraction constant $\lambda_2/\lambda_1$, i.e. for all $\vec x, \vec y \in \mathbb{R}^{n}$:
$$ d(f(\vec x),f(\vec y)) \le \frac {\lambda_2}{\lambda_1} d(\vec x,\vec y). $$
Moreover, $t = \frac{\log(d(\vec x_0, \vec v_1)/\eps)}{\log(\lambda_1/\lambda_2)}$
iterations suffice to have $\norm{\vec x_t - \vec v_1}_2 \le d(\vec x_t, \vec v_1) \le \eps$.
\end{proposition}

\begin{proof}
For any vector $\vec x$, it holds that $\langle A \vec x, \vec v_1 \rangle = \lambda_1 \langle \vec x, \vec v_1 \rangle.$
We have that
\begin{align*}
  d(f(\vec x),f(\vec y)) \quad &= \norm{\frac {A \vec x} {\langle A \vec x, \vec v_1 \rangle} - \frac {A \vec y} {\langle A \vec y, \vec v_1 \rangle}}_2 \\
  &= \frac 1 {\lambda_1} \norm{A \left( \frac {\vec x} {\langle \vec x, \vec v_1 \rangle} - \frac {\vec y} {\langle \vec y, \vec v_1 \rangle} \right) }_2\\
  &\le \frac {\lambda_1} {\lambda_1} \norm{ \frac {\vec x} {\langle \vec x, \vec v_1 \rangle} - \frac {\vec y} {\langle \vec y, \vec v_1 \rangle}  }_2 \quad \quad= \frac {\lambda_2}{\lambda_1} d(\vec x,\vec y)
\end{align*}
where the inequality is true as the vector $\frac {\vec x} {\langle \vec x, \vec v_1 \rangle} - \frac {\vec y} {\langle \vec y, \vec v_1 \rangle}$ is perpendicular to the principal eigenvector $\vec v_1$.
This shows that $f$ is contracting with respect to $d$ as required.

To convert a bound on the $d$ metric to a bound on the error with respect to the $\ell_2$ norm, we can see that at every step $t>0$, $\norm{\vec x_t}_2 = 1$.
If at some step $t>0$, it holds that $d(\vec x_t, \vec v_1) \le \varepsilon$, we get
$$\varepsilon^2 \ge d(\vec x_t, \vec v_1)^2 = \norm{\frac {\vec x_t} {\langle \vec x_t, \vec v_1 \rangle} - \vec v_1}^2_2 = \langle \vec x_t, \vec v_1 \rangle^{-2} - 1 \Rightarrow \langle \vec x_t, \vec v_1 \rangle \ge (1+\varepsilon^2)^{-1/2}.$$
This implies that
$\norm{\vec x_t - \vec v_1}_2^2 = 2(1 - \langle \vec x_t, \vec v_1 \rangle) \le 2\left(1-(1+\varepsilon^2)^{-1/2} \right) \le \varepsilon^2$.
This guarantees that bounding the norm $d$ by $\varepsilon$ implies a bound of $\varepsilon$ on the $\ell_2$ norm between the principal eigenvector and the current iterate $\vec x_t$.

Using these observations and following the same approach as in Corollaries \ref{cor:converse1}-\ref{cor:converse2} we get the required bound on the number of iterations.
\end{proof}

Notice, that the definition of the metric in Proposition~\ref{prop:powerContr} depends on the principal eigenvector but not on any of the other eigenvectors. When applied to show global convergence of
Markov chains, the principal eigenvector corresponds to the stationary distribution. For a symmetric Markov chain whose stationary distribution is uniform Proposition~\ref{prop:powerContr}
implies that the iterations are contracting directly with respect to the $\ell_2$ norm.

\section{Banach is Complete for $\CLS$}
\label{sec:cls}

As discussed in Section~\ref{sec:intro}, the complexity class $\CLS$ was defined in~\cite{DP11} to capture problems in the intersection of $\PPAD$ and $\PLS$, such as P-matrix LCP, mixed Nash equilibria of congestion and multi-player coordination games, finding KKT points, etc. It also contains computational variants of finding fixed points whose existence is guaranteed by Banach's fixed point theorem. In this section, we close the circle by proposing two variants of Banach fixed point computation that are both $\CLS$-complete. Our $\CLS$ completeness results are obtained by making our proof of Theorem~\ref{th:converse} constructive. We start with a formal definition of $\CLS$, which is defined in terms of the problem $\clocal$.

%


\begin{definition} \label{def:cls}
  \em
  $\clocal$ takes as input two functions $f : [0, 1]^3 \to [0, 1]^3$, $p : [0, 1]^3 \to [0, 1]$, both represented as arithmetic circuits, and two rational positive constants $\eps$ and $\lambda$. The desired output is any
of the following:
  \begin{CompactEnumerate}[label=(CO\arabic*)]
    \item a point $x \in [0, 1]^3$ such that $p(f(x)) \ge p(x) - \eps$. \label{cLocalO1}
    \item two points $x, x' \in [0, 1]^3$ violating the $\lambda$-Lipschitz continuity of $f$, i.e. \\
          $|f(x) - f(x')|_1 > \lambda |x - x'|_1$. \label{cLocalO2}
    \item two points $x, x'$ violating the $\lambda$-Lipschitz continuity of $p$, i.e. \\ $|p(x) - p(x')| > \lambda |x - x'|_1$. \label{cLocalO3}
  \end{CompactEnumerate}
  The class $\CLS$ is the set of search problems that can be reduced to $\clocal$.
\end{definition}

\begin{remark}
As discussed in~\cite{DP11}, both the choice of domain $[0,1]^3$ and the use of $\ell_1$ norm in the definition of the above problem are not crucial, and high-dimensional polytopes as well as other $\ell_p$ norms can also be used in the definition without any essential effect to the complexity of the problem. Moreover, instead of the functions $f$ and $p$ being provided in the input as arithmetic circuits there is a canonical way to provide them in the input as binary circuits that define the values of $f$ and $p$ on all points of some finite bit complexity, and (implicitly) extend to the full domain via continuous interpolation. In this way, we can syntactically guarantee the Lipschitz continuity of both $f$ and $p$ and can remove~\ref{cLocalO2} and~\ref{cLocalO3} from the above definition. For more details, please see~\cite{DP11},~\cite{DGP09} and~\cite{EtessamiY07}. This remark applies to all definitions in this section.
\end{remark}
\smallskip

  The variant of Banach's theorem that is known to belong to $\CLS$ is {\sc Contraction Map}, defined as follows:
\begin{definition}[\cite{DP11}] \label{def:contraction map}
  \em
  {\sc Contraction Map} takes as input a function $f : [0, 1]^3 \to [0, 1]^3$ represented as an arithmetic circuit and three rational positive constants $\eps$, $\lambda$, $c < 1$. The desired output is any of the following (where $d$ represents Euclidean distance):
  \begin{CompactEnumerate}[label=(O\alph*)]
    \item a point $x \in [0, 1]^3$ such that $d(x, f(x)) \le \eps$ \label{contrO1}
    \item two points $x, x' \in [0, 1]^3$ disproving the contraction of $f$ w.r.t. $d$ with constant $c$, i.e.\\
          $d(f(x), f(x')) > c \cdot d(x, x') $ \label{contrO3}
    \item two points $x, x' \in [0, 1]^3$ disproving the $\lambda$-Lipschitz continuity of $f$, i.e. \\
          $|f(x) - f(x')|_1 > \lambda |x - x'|_1$. \label{contrO4}
  \end{CompactEnumerate}
\end{definition}
\noindent {\sc Contraction Map} targets fixed points whose existence is guaranteed by Banach's fixed point theorem when $f$ is a contraction map with respect to the Euclidean distance. However, it doesn't capture the full generality of Banach's theorem, since the latter can be applied to any complete metric space. We thus define a more general problem, $\Banachh$ that: (i) still lies inside CLS, (ii) captures the generality of Banach's theorem, (iii) and in fact tightly captures the complexity of the class CLS, by being CLS-complete. This problem is defined as follows:
\begin{definition}\label{def:banachh}
  \em
  $\Banachh$ takes as input two functions $f : [0, 1]^3 \to [0, 1]^3$ and $d : [0, 1]^3 \times [0, 1]^3 \to \reals$ represented as arithmetic circuits, where $d$ is promised to be a metric that is topologically equivalent to the Euclidean distance and satisfy that $([0,1]^3,d)$ is a complete metric space, and three rational positive constants $\eps$, $\lambda$, $c < 1$. The desired output is any of the following:
  \begin{CompactEnumerate}[label=(O\alph*)]
    \item a point $x \in [0, 1]^3$ such that $d(x, f(x)) \le \eps$ \label{contrBanachhO1}
    \item two points $x, x' \in [0, 1]^3$ disproving the contraction of $f$ w.r.t. $d$ with constant $c$, i.e.\\
          $d(f(x), f(x')) > c \cdot d(x, x') $ \label{contrBanachhO3}
    \item two points $x, x' \in [0, 1]^3$ disproving the $\lambda$-Lipschitz continuity of $f$, i.e. \\
          $|f(x) - f(x')|_1 > \lambda |x - x'|_1$. \label{contr:BanachhO4}
    \item four points $x_1, x_2, y_1, y_2 \in [0, 1]^3$ with $x_1 \neq x_2$ and $y_1 \neq y_2$ disproving the 
          $\lambda$-Lipschitz continuity of $d(\cdot, \cdot)$, i.e.           
          $|d(x_1, x_2) - d(y_1, y_2)| > \lambda \left( \abs{x_1 - y_1}_1 + \abs{x_2 - y_2}_1 \right)$. \label{contr:BanachhO5}
  \end{CompactEnumerate}
\end{definition}

\begin{remark} \label{rem:continuityOfMetric}
We remark that $\Banachh$ is tightly related to {\sc Contraction Map} defined above, with the following differences. First, instead of Euclidean distance, the metric with respect to which $f$  is purportedly contacting is provided as part of the input and it is promised to be a metric. Second, we need to add an extra type of accepted solution \ref{contr:BanachhO5}, which is a violation of the Lipschitz property of that metric. This is necessary to guarantee that the above problem has a solution of polynomial length for any possible input, and in particular is needed to place the above problem in CLS. (It is not needed for the CLS-hardness.)
\end{remark}



Our main result is the following:

\begin{theorem} \label{th:cBanach2}
   $\Banachh$ is $\CLS$-complete.
\end{theorem}

  We give here a sketch of the proof of Theorem \ref{th:cBanach2} and we present the full proof in Appendix 
\ref{app:clsProofs}.

\begin{proof}[Proof Sketch.]
  Since the inclusion to $\CLS$ is a simple argument very similar to the argument from \cite{DP11} that shows that 
{\sc Contraction Map} belongs to $\CLS$, we focus here on the hardness proof.

  \medskip
  We are given two functions $f : [0, 1]^3 \to [0, 1]^3$, $p : [0, 1]^3 \to [0, 1]$ and we want to find a contraction
$d : [0, 1]^3 \times [0, 1]^3 \to \reals$ such that $f$ is a contraction map with respect to $d$ and the points where 
$p(f(x)) \ge  p(x) - \eps$ are approximate fixed points of $f$ with respect to $d$.

  The inspiration of this proof is to make the proof of Theorem \ref{th:converse} constructive in polynomial time. We 
therefore follow the steps of the proof sketch of Theorem \ref{th:converse} as presented in Section \ref{sec:BanachConverse}.

\smallskip
\noindent \textbf{Step I.} Since we don't have the strong requirement of Theorem \ref{th:converse} to output a metric that
is topologically equivalent with some given metric we can use in place of $d_M$ any metric $d'$ such that $f$ is non-expanding 
with respec to $d'$. Hence we can easily observe that the \textit{discrete metric} can be used as $d_M$.

\smallskip 
\noindent \textbf{Step II.} The construction of Theorem \ref{th:converse} uses in the definition of $d(x, y)$ the number 
of times $n(x)$, that we have to apply $f$ on $x$ in order for $f^{[n(x)]}(x)$ to come $\eps$-close to the fixed point $x^*$ of $f$. 
Of course $n(x)$ is not a quantity that can be computed in polynomial time. Instead we show that it suffices to use an upper 
bound on $n(x)$ which we can get using the potential function, namely $p(x) / \eps$. Of course the operations that we are allowed 
to use to describe $d$ as an arithmetic circuit are limited and this step appears to need more expressive power that the simple arithmetic operations that we are allowed to use. We give a careful construction that bypasses these difficulties and completes this step of the proof.

\smallskip
\noindent \textbf{Steps III.} This step of Theorem \ref{th:converse} is highly non-constructive and hence we cannot hope to 
replicate it in polynomial time. But we prove that our carefully designed metric already has the triangle inequality and hence 
the transitive closure step is not necessary.

\smallskip
\noindent The last part of our proof is to show that the constructed circuit of $d$ is actually Lipschitz with a relatively 
small Lipschitz constant if the potential function $p$ is Lipschitz. That is, we have to show
that the circuit of $d$ does not need some time exponentially many bits with respect to the size of the circuits of $p$ and 
the magnitude of the constant $1/\eps$. Not suprisingly we observe that in order to succeed to this task we have to set 
approximately $c = 1 - \eps$. This is natural to expect, since if we could set a much lower contraction constant then we could 
find the approximate fixed point of $f$ in much less that $\poly(1/\eps)$ steps which cannot hold given that $\CLS \neq 
\text{FP}$.
\end{proof}

\section*{Acknowledgements}
The authors were supported by NSF CCF-1551875, CCF-1617730, CCF-1650733, and a Simons Graduate Research Fellowship.

  \bibliographystyle{alpha}
  \bibliography{ref}
  
\clearpage
\appendix
\section{Preliminaries}

\subsection{Basic Definitions} \label{sec:app:def}

\paragraph{Topological Spaces} Let $\Domain$ be a set and $\tau$ a collection of subsets of $\Domain$ with the following properties.
  \begin{CompactEnumerate}[label=(\alph*)]
    \item The empty set $\emptyset \in \tau$ and the space $\Domain \in \tau$.
    \item If $U_a \in \tau$ for all $a \in A$ then $\bigcup_{a \in A} Ua \in \tau$.
    \item If $U_j \in \tau$ for all $1 \le j \le n \in \nats$, then $\bigcap_{j = 1}^n U_j \in \tau$.
  \end{CompactEnumerate}
  Then we say that $\tau$ is a \textit{topology} on $\Domain$ and that $(\Domain, \tau)$ is a \textit{topological space}. We call \textit{open sets} the members of $\tau$.
Also a subset $C$ of $\Domain$ is called \textit{closed} if $\Domain \setminus C$ is an open set, i.e. belongs to $\tau$. Let $(\Domain, \tau)$ be a topological space and
$A$ a subset of $\Domain$. We write
  \begin{align*}
    \label{eq:intClos}
    \Int(A) & = \bigcup \{U \in \tau \mid U \subseteq A\} \\
    \Clos(A) & = \bigcap \{U \text{ closed} \mid A \subseteq U\}
  \end{align*}
\noindent and we call $\Clos(A)$ the \textit{closure} of $A$ and $\Int(A)$ the \textit{interior} of $A$. We now give a basic lemma without proof. A proof can be
found in \cite{korner10}.

  \begin{lemma}
    \label{lem:interior}
    \begin{CompactEnumerate}[label=(\roman*)]
      \item $\Int(A) = \{ x \in A \mid \exists U \in \tau \text{ with } x \in U \subseteq A\}$.
      \item $\Clos(A) = \{ x \in \Domain \mid \forall U \in \tau \text{ with } x \in U, \text{ we have } U \cap A \neq \emptyset\}$.
    \end{CompactEnumerate}
  \end{lemma}

  \paragraph{Metric Spaces}

      The \textit{diameter} of a set $W \subseteq \Domain$ according to the metric $d$ is defined as
      \[ \diam{d}{W} = \max_{x, y \in W} d(x, y) \]

    A metric space $(\Domain, d)$ is called \textit{bounded} if $\diam{d}{\Domain}$ is finite.

    \noindent We define $d_S : \Domain^2 \to \reals$ by
    \begin{equation*}
      d_S(x, y) = \left\{ \begin{split} 0 & ~~ x = y \\ 1 & ~~ x \neq y \end{split} \right.
    \end{equation*}
    then $d_S$ is called the \textit{discrete metric} on $\Domain$.

  \paragraph{Remark.} It is very easy to see that discrete metric is indeed a metric, i.e. it satisfies the conditions \ref{eq:dMetric1}-\ref{eq:dMetric4}. \\

      Let $(\Domain, d)$ and $(\Support, d')$ be metric spaces. A function $f : \Domain \to \Support$ is called \textit{continuous} if, given $x \in \Domain$ and
      $\eps > 0$, we can find a $\delta(x, \eps) > 0$ such that
      \[ d'(f(x), f(y)) < \eps \text{ whenever } d(x, y) < \delta(x, \eps) \]

      We say that a subset $E \subseteq \Domain$ is \textit{open} in $\Domain$ if, whenever $e \in E$, we can find a
      $\delta > 0$ (depending on $e$) such that
      \[ x \in E \text{ whenever } d(x, e) < \delta \]

      The next lemma connects the definition of open sets according to some metric with the definition of open sets in a topological space.

    \begin{lemma}
      If $(\Domain, d)$ is a metric space, then the collection of open sets forms a topology.
    \end{lemma}

      We define the \textit{open ball} of radius $r$ around $x$ to be $B(x, r) = \{ y \in \Domain | d(x, y) < r \}$.

    \paragraph{Closed Sets for Metric Spaces} Consider a sequence $(x_n)$ in a metric space $(\Domain, d)$. If $x \in \Domain$ and, given $\eps > 0$, we can find an
  integer $N \in \nats_1$ (depending maybe on $\eps$) such that
      \[ d(x_n, x) < \eps \text{ for all } n \ge N \]
  then we say that $x_n \rightarrow x$ as $n \rightarrow \infty$ and that $x$ is the \textit{limit} of the sequence $(x_n)$.\\
  A set $G \subseteq \Domain$ is said to be \textit{closed} if, whenever $x_n \in G$ and $x_n \rightarrow x$ then $x \in G$. A proof of the following lemma can be
  found in \cite{korner10}.

    \begin{lemma}
        Let $(\Domain, d)$ be a metric space and $A$ a subset of $\Domain$. Then $\Clos(A)$ consists of all those $x \in \Domain$ such that we can find $(x_n)$ with
      $x_n \in A$ with $d(x_n, x) \rightarrow 0$.
    \end{lemma}

  \noindent We define the \textit{closed ball} of radius $r$ around $x$ to be $\bar{B}(x, r) = \{ y \in \Domain | d(x, y) \le r \}$. \\
  A subset $G$ of a metric space $(\Domain, d)$ is called \textit{compact} if $G$ is closed and every sequence in $G$ has a convergent subsequence. A metric space
  $(\Domain, d)$ is called compact if $\Domain$ is compact, \textit{locally compact} if for any $x \in \Domain$, $x$ has a neighborhood that is compact and
  \textit{proper} if every closed ball is compact.

    \paragraph{Complete Metric Spaces} We say that a sequence $(x_n)$ in $\Domain$ is \textit{Cauchy sequence} (or \textit{$d$-Cauchy sequence} if the distance
      metric is not clear from the context) if, given $\eps > 0$, we can find $N(\eps) \in \nats_1$ with
      \[ d(x_n, x_m) < \eps \textit{ whenever } n, m \ge N(\eps) \]

      \noindent A metric space $(\Domain, d)$ is \textit{complete} if every Cauchy sequence converges.

      Two metrics $d$, $d'$ of the same set $\Domain$ are called \textit{topologically equivalent} (or just \textit{equivalent}) if for every sequence $(x_n)$ in
      $\Domain$, $(x_n)$ is $d$-Cauchy sequence if and only if it is $d'$-Cauchy sequence.

      \begin{definition} \label{def:continuityM}
        Let $(\Domain, d)$ be a metric spaces. A function $f : \Domain \to \Domain$ is called \textit{continuous with repect to $d$}, if given $x \in \Domain$ and
        $\eps > 0$, we can find a $\delta(x, \eps)$ such that
        \[ d'(f(x), f(y)) < \eps \text{ whenever } d(x, y) < \delta(x, \eps) \]
      \end{definition}

   \paragraph{Lipschitz Continuity} Let $(\Domain, d)$ and $(\Support, d')$ be metric spaces. A function $f : \Domain \to \Support$ is \textit{Lipschitz continuous}
    (or \textit{$(d, d')$-Lipschitz
    continuous} if the distance metric is not clear from the context or \textit{$d$-Lipschitz continuous} if $d = d'$) if there exists a positive constant $\lambda \in \reals_+$ such that for
    all $x, y \in \Domain$
    \[ d'(f(x), f(y)) \le \lambda d(x, y) \]

    \begin{lemma}
      \label{lem:Lipschitz}
      If a function $f : \Domain \to \Support$ is Lipschitz continuous then it is continuous.
    \end{lemma}

    \begin{definition}
      \label{def:contractionM}
        Let $(\Domain, d)$ and $(\Support, d')$ be metric spaces. A function $f : \Domain \to \Support$ is \textit{contraction} (or \textit{$(d, d')$-contraction} or \textit{$d$-contraction} if $d = d'$)
      if there exists a positive constant $1 > c \in \reals_+$ such that for all $x, y \in \Domain$
      \[ d'(f(x), f(y)) \le c d(x, y) \]
      If $c = 1$ then we call $f$ \textit{non-expansion}.
    \end{definition}
    \noindent A \textit{fixed point} of a selfmap $f$ is any point $x^* \in \Domain$ such that $f(x^*) = x^*$.

\subsection{Introduction to Power Method} \label{sec:app:power}

      Let $A \in \reals^{n \times n}$. Recall that if $q$ is an eigenvector for $A$ with eigenvalue $\lambda$, then $A q = \lambda q$, and in general,
    $A^k q = \lambda^k q$ for all $k \in \nats$. This observation is the foundation of the \textit{power iteration method}.

      Suppose that the set $\{q_i\}$ of unit eigenvectors of $A$ forms a basis of $\reals^n$, and has corresponding set of real eigenvalues $\{\lambda_i\}$ such that
    $\abs{\lambda_1} > \abs{\lambda_2} > \dots > \abs{\lambda_n}$. Let $v_0$ be an arbitrary initial vector, not perpendicular to $q_1$, with $\norm{v_0} = 1$. We
    can write $v_0$ as a linear combination of the eigenvectors of $A$ for some $c_1 \dots, c_n \in \reals$ we have that
    \[ v_0 = c_1 q_1 + c_2 q_2 + \dots + c_n q_n \]

    \noindent and since we assumed that $v_0$ is not perpendicular to $q_1$ we have that $c_1 \neq 0$.

    \noindent Also
      \[ A v_0 = c_1 \lambda_1 q_1 + c_2 \lambda_2 q_2 + \dots + c_n \lambda_n q_n \]

    \noindent and therefore
    \begin{align*}
      A v_k & = c_1 \lambda_1^k q_1 + c_2 \lambda_2^k q_2 + \cdots + c_n \lambda_n^k q_n \\
            & = \lambda_1^k \left( c_1 q_1 + c_2 \left( \frac{\lambda_2}{\lambda_1} \right)^k q_2 + \cdots + c_n \left( \frac{\lambda_n}{\lambda_1} \right)^k q_n \right)
    \end{align*}

      Since the eigenvalues are assumed to be real, distinct, and ordered by decreasing magnitude, it follows that
      \[ \lim_{k \to \infty} \left( \frac{\lambda_i}{\lambda_1} \right)^k = 0 \]

      So, as $k$ increases, $A^k v_0$ approaches $c_1 \lambda_1^k q_1$, and thus for large values,
      \[ \frac{A^k v_0}{\norm{A^k v_0}} \rightarrow q_1 \text{ as } k \to \infty \]

      The power iteration method is simple and elegant, but suffers some  drawbacks. Except from a measure $0$ of initial conditions, the method returns a single eigenvector, corresponding to the eigenvalue of largest magnitude.  In addition, convergence is only guaranteed if the eigenvalues are distinct—in particular, the two
    eigenvalues of largest absolute value must have distinct magnitudes. The rate of convergence primarily depends upon the ratio of these magnitudes, so if the two
    largest eigenvalues have similar sizes, then the convergence will be slow.

      In spite of its drawbacks, the power method is still used in many applications, since it works well on large, sparse matrices when only a single eigenvector is needed.  However, there are other methods overcoming some of the issues with the power iteration method.

\section{Proof of Theorem \ref{th:converse}} \label{sec:app:mainproof}

\subsection{Meyer's Construction and Our Contribution}
  The proof of Theorem \ref{th:converse}, follows the construction of \cite{Meyers1967}. We give a complete step by step 
description of this construction. For every step of this construction we explain what it was already proven by Meyers and 
we additionally prove some properties that are needed in order to satisfy our additional condition \eqref{eq:converseCond2}.

  The construction of $d_{c, \eps}$ starts with an open neighborhood of $x^*$ with some desired properties. Meyers starts with 
an arbitrary open neighborhood $W$ such that $f(W) \subset W$, whereas we also need that $\diam{d}{W} \le \eps$.

\begin{lemma} \label{lem:neighborhood}
  There exists an open neighborhood $W$ of $x^*$ such that
  \begin{subequations} \label{eq:neighborhood}
      \begin{align}
        & f(W) \subseteq W \label{eq:neighborhood2} \\
        & \diam{d}{W} \le \eps \label{eq:neighborhood3}
      \end{align}
    \end{subequations}
\end{lemma}

\begin{pproof}
%
  From the hypothesis of the theorem there exists an open neighborhood $U$ such that $f^{[n]}(U) \rightarrow \{x^*\}$. This implies that any open subset $V$ of $U$ satisfies $f^{[n]}(V) \to \{x^*\}$. Therefore we can choose a $V = \Int(\bar{B}(x^*, \eps))$
such that $\diam{d}{V} \le \eps$ and $f^{[n]}(V) \to \{x^*\}$. For simplicity of the notation we just assume refer to $V$ as $U$ and so $\diam{d}{U} \le \eps$.

  Starting from $U$ we prove the existence of $W$. For this, we will prove that there exists an open neighborhood $W$
of $x^*$ such that $f(W) \subset W$ and $W \subset U$. The latter implies $f^{[n]}(W) \rightarrow \{x^*\}$ and $\diam{d}{W} \le \eps$.

  Since $f^{[n]}(U) \rightarrow \{x^*\}$, there is an integer $k$ such that $f^{[k]}(U) \subseteq U$. Let
  \[ W = \bigcap_{j = 0}^{k - 1} f^{[-j]}(U) \subseteq U \]

  \noindent Then for any $x \in W$ and for any $1 \le j \le k - 1$ it holds that $x \in f^{[-j]}(U)$ and thus $f(x) \in f^{[-(j - 1)]}(U)$. Moreover $x \in U$, so that
$f^{[k]}(x) \in f^{[k]}(U) \subset U$ and thus $f(x) \in f^{[-(k - 1)]}(U)$. Hence $x \in W$ implies $f(x) \in W$, which was to be shown. The diameter of $W$ can be bounded by the diameter of $U$ and hence is less that $\eps$.
\end{pproof}

\vspace{4pt}

  We now proceed to the main line of the proof. The construction follows three steps:
\begin{CompactEnumerate}[label=\Roman*.]
  \item We first construct a metric $d_M$, which is topologically equivalent to $d$, and with respect to which $f$ is non-expanding. It also holds that $d_M(x, y) \ge d(x, y)$ for
        all $x, y \in \Domain$ and therefore Property \eqref{eq:converseCond2} can be transferred from $d_M$ to $d$.
  \item Given $d_M$, we proceed to construct a ``distance'' function $\rho_{c, \eps}$, which satisfies~\eqref{eq:converseCond1} and all the metric properties except maybe for the triangle inequality. Moreover $\rho_{c, \eps}$ satisfies that $\rho_{c, \eps}(x, y) \ge d_M(x, y)$ if $\max\{d(x^*, x), d(x^*, y)\} \ge \eps$, and therefore \eqref{eq:converseCond2} is
        preserved.
  \item Given $\rho_{c, \eps}$, we construct the sought after metric $d_{c, \eps}$ by taking it equal to the $\rho_{c, \eps}$-geodesic distance. Given the properties of $\rho_{c,\eps}$ and the definition of $d_{c, \eps}$, we can prove that $d_{c,\eps}$ is a metric and Properties~\eqref{eq:converseCond1} and~\eqref{eq:converseCond2} hold.
\end{CompactEnumerate}


\paragraph{I. Construction of $\mathbf{d_M}$ \\}

  In the fist step of the construction we define an metric $d_M$ as
  \[ d_M(x, y) = \sup_{n \in \nats}\{ d(f^{[n]}(x), f^{[n]}(y)) \}\]
\noindent and we show that $f$ is non-expanding with respect to $d_M$.

\begin{lemma}[\cite{Meyers1967}] \label{lem:proofMeyersLemmaI}
  For the $d_M$ function defined above we have that:
\begin{Enumerate}
  \item $d_M$ is well defined and satisfies all the metric properties (see Definition \ref{def:dMetric}).
  \item $d_M$ is topologically equivalent with $d$.
\end{Enumerate}
\end{lemma}

  The proof of Lemma \ref{lem:proofMeyersLemmaI} can be found in Section \ref{sec:meyerssOmittedProofs} where for completeness 
we keep the proofs that were already proved by Meyers's \cite{Meyers1967}.

\medskip
  For our purposes we also observe that by the definition of $d_M$ the following holds
\[ d_M(x, y) \ge d(x, y) \]
\noindent and hence if $d_M(x, y) \le \eps$ then also $d(x, y) \le \eps$, therefore $d_M$ satisfies \eqref{eq:converseCond2}.

\paragraph{II. Construction of $\mathbf{\rho_c}$ \\}

  We begin by defining $K_n$ to be the closure of $f^n(W)$ for $n \ge 0$, in particular we have that $K_0 = W$ and hence by 
Lemma \ref{lem:neighborhood} we have that $\diam{d}{K_0}$. Also we define $K_{(-n)} = f^{[-n]}(K_0)$, so that our assumption 
$f^{[n]}(W) \to \{x^*\}$ implies
  \begin{equation} \label{eq:meyersProof2}
    K_n \to \{x^*\} ~~~~~ \text{as }~ n \to \infty.
  \end{equation}
\noindent For $x \in K_0 \setminus \{x^*\}$ we set $n(x) = \max_{x \in K_n}\{n\} \ge 0$. The fact that $n(x)$ is finite is 
guaranteed by (\ref{eq:meyersProof2}). To see this assume that there is an infinite sequence $n_1, n_2, \dots$ such that 
$x \in K_{n_i}$ which implies that $x \in \bigcap_{i = 1}^{\infty} K_i$ which definitely contradicts \eqref{eq:meyersProof2}. 
We define also $n(x^*) = \infty$, and for $x \in \Domain \setminus K_0$ set 
$n(x) = - \min_{f^{[m]}(x) \in K_0} \{m\} = \max_{x \in K_n}\{n\} < 0$ which again is finite because of condition 2. Let
$\kappa(x, y) = \min\{n(x), n(y)\}$, we define $\rho_c$ to be
  \[ \rho_c(x, y) = c^{\kappa(x, y)} d_M(x, y). \]

\begin{lemma}[\cite{Meyers1967}] \label{lem:proofMeyersLemmaII}
  For the $\rho_c$ function defined above we have that:
\begin{Enumerate}
  \item $\rho_c$ is well defined and satisfies the metric properties (see Definition \ref{def:dMetric}), except the      
        triangle inequality \ref{eq:dMetric4}.
  \item $f$ is a contraction map with respect to $\rho_c$ with contraction constant $c$.
          \begin{equation} \label{eq:rhoc0}
            \rho_c(f(x), f(y)) \le c \cdot \rho_c(x, y).
          \end{equation}
\end{Enumerate}
\end{lemma}

  The proof of Lemma \ref{lem:proofMeyersLemmaII} is almost immediate from the definition of $\rho_c$, but for a detailed 
explanation we refer to the initial proof by Meyers \cite{Meyers1967}.

\paragraph{III. Construction of $\mathbf{d_c}$ \\}

  In this last step what we do is that we assign the distance between two points to be the length of the 
shortest path that connects these two points, with the lengths computed according to $\rho_c$. Then the 
distance satisfies the triangle inequality because of the shortest path property.

  Formally, denote by $S_{xy}$ the set of chains $s_{xy} = (x = x_0, x_1, \dots, x_m = y)$ from $x$ to $y$ with 
associated lengths $L_c(s_{xy}) = \sum_{i = 1}^m \rho_c(x_i, x_{i - 1})$. We define
  \begin{equation} \label{eq:meyersProof3}
    d_c(x, y) = \inf\{L_c(s_{xy}) \mid s_{xy} \in S_{xy}\}.
  \end{equation}

\begin{lemma}[\cite{Meyers1967}] \label{lem:proofMeyersLemmaIII}
  For the $d_c$ function defined above we have that:
\begin{Enumerate}
  \item $d_c$ is well defined and satisfies all the metric properties (see Definition \ref{def:dMetric}).
  \item $f$ is a contraction map with respect to $d_c$ with contraction constant $c$.
          \begin{equation} \label{eq:rhoc1}
            d_c(f(x), f(y)) \le c \cdot d_c(x, y).
          \end{equation}
  \item $d_c$ is topologically equivalent with $d$ and hence $(\Domain, d_c)$ is a complete metric space.
\end{Enumerate}
\end{lemma}

  The proof of Lemma \ref{lem:proofMeyersLemmaIII} can be found in Section \ref{sec:meyerssOmittedProofs}. 
  
  We know need to prove two lemmas to prove that Meyers's construction also satisfies \eqref{eq:converseCond2}.

\begin{lemma} \label{lem:proofOursLemmaIII1}
  Consider any $x \neq x^*$ and $y \neq x, x^*$ and assume that $y \notin K_0$, then
\begin{equation} \label{eq:ourProof1}
  d_c(x, y) \ge \min\{d_M(x, y), d_M(x, K_{0})\} > 0.
\end{equation}
\end{lemma}

\begin{prevproof}{Lemma}{lem:proofOursLemmaIII1}
  By definition any chain $s_{x y}$ either lies in $\Domain \setminus K_{0}$ entirely, or has a last link which leaves $
K_0$. If $s_{x y}$ lies in $\Domain \setminus K_{0}$ entirely then $n(x), n(y) < 0$ and hence $d_c(x, y) \ge d_M(x, y)$. 
Otherwise we consider the last link that leaves $K_0$ and we have that the length between $x$ and $y$ according to $d_c$
is greater than the distance with respect to $d_M$ of $x$ from $K_0$ which gives that
\[ d_M(x, K_0) \le d_c(x, y). \]
\end{prevproof}

  The final step is to prove (\ref{eq:converseCond2}). 
  
\begin{lemma} \label{lem:proofOursLemmaIII2}
  \[ \forall x,y \in \Domain: d_{c, \eps}(x, y) \le \eps \implies \min\{d(x^*, x), d(x^*, y), d(x, y)\} \le 2 \eps. \]
\end{lemma}

\begin{prevproof}{Lemma}{lem:proofOursLemmaIII2}
  Let $A = \diam{d}{\bar{B}(x^*, 2 \eps)}$ and without loss of generality $d(x, x^*) \ge d(y, x^*)$.

  If either $d_M(x, x^*) \le \eps$ or $d_M(y, x^*) \le \eps$ then we are done since as we have seen in the 
construction of $d_M$, $d_M(x, y) \ge d(x, y)$, thus either $d(x, x^*) \le \eps$ or $d(y, x^*) \le \eps$ and 
(\ref{eq:converseCond2}) is satisfied. So we may assume that $d(x, x^*) \ge \eps$ and $d(y, x^*) \ge \eps$.
Therefore $x, y \in \Domain \setminus K_0$ and which translates to $n(x), n(y) < 0$. So now using Lemma 
\ref{lem:proofOursLemmaIII1} and we get
  \begin{equation} \label{eq:meyersProof11}
    d_c(x, y) \ge \min\{ d_M(x, y), d_M(x, K_0) \}.
  \end{equation}

  	Now we consider two cases according to the value of $d_M(x, K_0)$. If $d_M(x, K_0) \ge \eps$ then
  \[ d_c(x, y) \le \eps \implies d_M(x, y) \le \eps \le A \implies d(x, y) \le \eps. \]
 \noindent Otherwise if $d_M(x, K_0) \le \eps$ then $d(x, K_0) \le \eps$ and by triangle inequality 
$d(x, x^*) \le 2 \eps$. By our assumption for the relative position of $x$ and $y$ we also get 
$d(y, x^*) \le 2 \eps$ and therefore $x, y \in \bar{B}(x^*, 2 \eps)$. Thus, 
$d(x, y) \le \diam{d}{\bar{B}(x^*, 2 \eps)}$. 
\end{prevproof}
\hfill \qed

\subsection{Omitted Proof of Lemmas proven by Meyers in \cite{Meyers1967}} \label{sec:meyerssOmittedProofs}

\begin{prevproof}{Lemma}{lem:proofMeyersLemmaI}
  The fact that this maximum is finite can be proved using the condition 2. of the theorem. Indeed, since $d(f^{[n]}(x), x^*) \to 0$ and
$d(f^{[n]}(y), x^*) \to 0$, for any $\delta > 0$ there is a number $N \in \nats$ such that $d(f^{[n]}(x), x^*) \le \delta$ and $d(f^{[n]}(y), x^*) \le \delta$ for
all $n > N$. Now if let $\delta = d(x, y)$ we get that $\max_{n \ge N}\{ d(f^{[n]}(x), f^{[n]}(y))\} \le d(x, y)$ and therefore
$\max_{n \in \nats}\{ d(f^{[n]}(x), f^{[n]}(y)) \} = \max_{0 \le n \le N}\{ d(f^{[n]}(x), f^{[n]}(y)) \}$. Hence the maximum has a finite value. Observe also that by definition it holds that
  \[ d_M(f(x), f(y)) \le d_M(x, y) \]
\noindent and hence $f$ is a non-xpansion according to $d_M$. It only remains to prove that $d_M$ satisfies the
properties of a metric function. The positive definiteness and symmetry of $d_M$ follow immediately from the
corresponding properties of $d$. The fact that $d_M(x, y) \neq 0$ for $x \neq y$ follows from the fact that
$d(x, y) \le d_M(x, y)$, which follows directly from the definition of $d_M$ since $f^{[0]}(x) = x$. It remains 
to prove the triangle inequality. For this we observe that by the definition of $d_M$ and usign the fact that 
the maximum in the definition of $d_M$ for any $x, y \in \Domain$ is finite, there exists an $n \in \nats$ such
that
\begin{align*}
	d_M(x, z) & = d(f^{[n]}(x), f^{[n]}(z)) \le \\
            & \le d(f^{[n]}(x), f^{[n]}(y)) + d(f^{[n]}(y), f^{[n]}(x)) \le \\
            & \le d_M(x, y) + d_M(y, z).
\end{align*}

  \noindent Thus $d_M$ is indeed a metric. We now show that $d_M$ is topologically equivalent to $d$. From the
inequality $d(x, y) \le d_M(x, y)$ it follows that any $d_M$-convergent sequence is also $d$-convergent, with
the same limit point. To prove the implication in the opposite direction, note that condition 2. of the 
hypotheses of the theorem implies the existence for each $\eta > 0$ of an $N$ such that
  \[ \diam{d}{f^{[n]}(W)} < \eta \text{~~~ for } n > N. \]
\noindent For each $x \in \Domain$, it follows from 2. that
\begin{equation} \label{eq:nudef}
  \nu(x) = \min_{n \in \nats, f^{[n]}(x) \in W}\{n\}
\end{equation}
\noindent is finite. Since $f$ is continuous, there is an $\delta > 0$ so small that $d(x, y) < \delta$ implies
\begin{equation} \label{eq:meyersProof1}
  f^{[\nu(x)]}(y) \in W \text { and } d(f^{[j]}(x), f^{[j]}(y)) < \eta \text{   for } 0 \le j \le N + \nu(x).
\end{equation}
By \eqref{eq:neighborhood2} we have $f^{[n + N + \nu(x)]}(x) \in f^{[n + N]}(W)$ and 
$f^{[n + N + \nu(x)]}(y) \in f^{[n + N]}(W)$ for all $n > 0$, so that the (\ref{eq:meyersProof1}) implies
  \[ d(f^{[j]}(x), f^{[j]}(y)) < \eta \text{~~~ for } j > N + \nu(x). \]
\noindent Thus $d(x, y) \le \delta$ implies $d_M(x, y) \le \eta$. This shows that a sequence which is 
$d$-convergent to $x$ is also $d_M$-convergent to $x$, completing the proof of topological equivalence. Finally
since $d$ and $d_M$ are topologically equivalent and $d$ is complete for $\Domain$ it follows that $d_M$ is 
also complete for $\Domain$.
\end{prevproof}

\begin{prevproof}{Lemma}{lem:proofMeyersLemmaIII}
  We will prove that $d_c$ is the desired metric. That $f$ is a contraction with constant $c$ with respect to 
$d_c$ follows by applying (\ref{eq:rhoc1}) to the links $[x_{i - 1}, x_i]$ of any chain $s_{xy}$. Clearly $d_c$ 
is symmetric and $d_c(x, x) = 0$. The triangle law holds since following a $s_{xy}$ with a $s_{yz}$ yields a 
$s_{xz}$. It remains to show that it is positive definite.

  Consider any $x \neq x^*$ and $y \neq x$ and assume $n(x) \le n(y)$ without loss of generality. If 
$y \neq x^*$, any chain $s_{xy}$ either lies in $\Domain \setminus K_{n(y) + 1}$, or has a last link which 
leaves $K_{n(y) + 1}$, so that 
\begin{equation} \label{eq:meyersProof4}
    d_c(x, y) \ge c^{n(y)} \min\{d_M(x, y), d_M(x, K_{n(y) + 1})\} > 0.
  \end{equation}

\noindent The remaining case, $y = x^*$ is covered by
  \begin{equation} \label{eq:meyersProof5}
  d_c(x, y) \ge c^{n(x)} d_M(x, K_{n(x) + 1}) > 0.
  \end{equation}

\noindent Thus $d_c$ is a distance metric. We now have to prove that $d_c$ is equivalent to $d_M$. Let 
$B_{\nu} = \Domain \setminus f^{[-\nu]}(W)$ for $\nu \ge 0$, so that the definition of $\nu(x)$ 
(\ref{eq:nudef}) implies $d_M(x, B_{\nu(x)}) > 0$ and $n(x) \ge - \nu(x)$. For any $x \neq x^*$, if $y$ obeys
  \begin{equation} \label{eq:meyersProof6}
    d_M(x, y) < \delta(x) = \min\{ d_M(x, K_{n(x) + 1}), d_M(x, B_{\nu(x)}) \}
  \end{equation}
 \noindent then $n(x) \ge - \nu(x)$, so that (\ref{eq:meyersProof3}) and (\ref{eq:meyersProof4}), the last with
 $x$ and $y$ interchanged, imply
  \begin{equation} \label{eq:meyersProof7}
    c^{n(x)} d_M(x, y) \le d_c(x, y) \le \rho_c(x, y) \le c^{- \nu(x)} d_M(x, y).
  \end{equation}

 \noindent Now choose $k(x) > \max\{0, n(x)\}$ such that $z \in K_{k(x)}$ implies 
 $d_M(z, x^*) < d_c(x, x^*)/2$. Then $d_c(x, K_{k(x)}) \ge d_c(x, x^*)/2$, so that if $y$ obeys
  \begin{equation} \label{eq:meyersProof8}
    d_c(x, y) < d_c(x, x^*)/2
  \end{equation}
  \noindent then only chains disjoint from $K_{k(x)}$ need enter (\ref{eq:meyersProof3}), implying
  \begin{equation}
    \label{eq:meyersProof9}
    d_c(x, y) \ge c^{k(x)} d_M(x, y).
  \end{equation}
 \noindent In particular, if
  \[ d_c(x, y) < \min\{ d_c(x, x^*)/2, c^{k(x)} \delta(x) \} \]
  \noindent then with (\ref{eq:meyersProof8}) and (\ref{eq:meyersProof9}) this implies (\ref{eq:meyersProof6}) and hence (\ref{eq:meyersProof7}) applies. Thus
$d_c(x_n, x) \to 0$ whenever $d_M(x_n, x) \to 0$.

 	Now if $x = x^*$, note first that if $d_M(x^*, y) < d_M(x^*, B_0)$, then
  \begin{equation} \label{eq:meyersProof10}
		d_c(x^*, y) \le \rho_c(x^*, y) \le d_M(x^*, y).
  \end{equation}

 \noindent Also note that for any $\eta > 0$, $f^{[n]}(W) \to \{x^*\}$ guarantees an $N(\eta) > 0$ such that 
$d_M(x^*, z) < \eta/2$ for all $z \in K_{N(\eta)}$. Then $d_M(x^*, y) > \eta$ implies that 
$d_M(y, K_{N(\eta)}) \ge \eta / 2$ and thus that
\[ d_c(x^*, y) \ge d_c(K_{N(\eta)}, y) \ge c^{N(\eta)} \eta/2 \]

\noindent Hence $d_c(x_n, x^*) \to 0$ if and only if $d_M(x_n, x^*) \to 0$.

  To show that $d_M$-completeness is preserved, assume that $(x_n)$ is a $d_c$-Cauchy sequence and that 
$(X, d_M)$ is complete. If $(x_n)$ does not converge to $x^*$ then since $d_c$ and $d_M$ are equivalent, for
some $N \in \nats$ and all sufficiently large $n$, $n(x_n) < N$.

\noindent Now exactly as above choose $k((x_n)) = P > \max{0, N}$ such that $z \in K_{k((x_n))}$ implies
  \[ d_M(x^*, z) < \inf_{i \in \nats} \left\{ \frac{d_c(x_i, x^*)}{2} \right\} = \frac R 2 \]
\noindent then since $(x_n)$ is a Cauchy sequence there is an $i \in \nats_1$ such that
  \[ d_c(x_p, x_{p + j}) < \frac R 2 \]
\noindent for all $p > i$, and using (\ref{eq:meyersProof9}) with $k(x) = P$, we have
  \[ c^{-P} d_c(x_p, x_{p + j}) \ge d_M(x_p, x_{p + j}) \]

\noindent so that $(x_n)$ is a $d_M$-Cauchy sequence. Therefore since $(\Domain, d_M)$ is complete, the 
topological space $(\Domain, d_c)$ is complete too.
\end{prevproof}

\section{Proof of Theorem \ref{th:cBanach2}} \label{app:clsProofs}

On route to establishing the CLS-completeness of $\Banachh$, we will define an intermediate, syntactic problem $\Banach$, which is similar to $\Banachh$ except that the function $d$ given in the input is not promised to be a metric, and hence a violation of the metricity of $d$ is accepted as a solution.

\begin{definition}
  \em
  $\Banach$ takes as input two functions $f : [0, 1]^3 \to [0, 1]^3$, $d : [0, 1]^3 \times [0, 1]^3 \to \reals$, both represented as arithmetic circuits, and three rational positive constants $\eps$, $\lambda$, $c < 1$. The desired output is any of the following:
  \begin{CompactEnumerate}[label=(O\alph*)]
    \item a point $x \in [0, 1]^3$ such that $d(x, f(x)) \le \eps$ \label{BanachO1}
    \item two points $x, x' \in [0, 1]^3$ disproving the contraction of $f$ w.r.t. $d$ with constant $c$, i.e.\\
          $d(f(x), f(x')) > c \cdot d(x, x') $ \label{BanachO3}
    \item two points $x, x' \in [0, 1]^3$ disproving the $\lambda$-Lipschitz continuity of $f$, i.e. \\
          $|f(x) - f(x')|_1 > \lambda |x - x'|_1$. \label{BanachO4}
    \item four points $x_1, x_2, y_1, y_2 \in [0, 1]^3$ with $x_1 \neq x_2$ and $y_1 \neq y_2$ disproving the $\lambda$-Lipschitz continuity of $d(\cdot, \cdot)$, i.e.
          $|d(x_1, x_2) - d(y_1, y_2)| > \lambda \left( \abs{x_1 - y_1} + \abs{x_2 - y_2} \right)$. \label{BanachO5}
    \item points $x, y, z \in [0,1]^3$ violating any of the metric properties of $d$ (\ref{eq:dMetric1}-\ref{eq:dMetric4} of Definition \ref{def:dMetric}). \label{BanachO6}
  \end{CompactEnumerate}
\end{definition}

  \smallskip

Notice that $\Banach$ is syntactic, namely for any input there exists a solution. We proceed to show that the problem is CLS-complete. 

\begin{theorem} \label{th:cBanach1}
   $\Banach$  is $\CLS$-complete. 
\end{theorem}

\begin{proof}[Proof of Theorem \ref{th:cBanach1}]

  We first show that $\Banach$ belongs to $\CLS$ even when we disallow \ref{BanachO6}. Starting from an instance $(f, d, \eps, \lambda, c)$ of $\clocal$ we create the following instance
  \[ ( f'(x) = f(x), p(x) = d(x, f(x)), \eps' = (1 - c) \cdot \eps, \lambda' = \lambda ) \]

\noindent Now we have to show that any output of the $\clocal$ with input $(f, p, \eps', \lambda)$ will give us a output of $\Banach$ with input $(f, d, \eps, \lambda, c)$.

\noindent \ref{cLocalO1} $\implies$ If $d(f(x), f(f(x))) > c \cdot d(x, f(x))$ then $(x, f(x))$ satisfies \ref{BanachO3}. Otherwise
  \begin{align*}
    p(f(x)) \ge p(x) - \eps' \Rightarrow d(f(x), f(f(x))) & \ge d(x, f(x)) - \eps' \Rightarrow \\
    c \cdot d(x, f(x)) \ge d(f(x), f(f(x))) & \ge d(x, f(x)) - \eps' \Rightarrow \\
    c \cdot d(x, f(x)) & \ge d(x, f(x)) - (1 - c) \cdot \eps \Rightarrow \\
    (1 - c) \cdot d(x, f(x)) & \le (1 - c) \cdot \eps \Rightarrow \\
    d(x, f(x)) & \le \eps
  \end{align*}

\noindent Therefore $x$ satisfies \ref{BanachO1} and therefore is a solution of $\Banach$.

\noindent \ref{cLocalO2} $\implies$ \ref{BanachO4}.

\noindent \ref{cLocalO3} $\implies$ Without loss of generality let $\norm{x - f(x)} \le \norm{y - f(y)}$. If $x = f(x)$ then if $d(x, f(x)) = 0$ we immediately satisfy \ref{BanachO1} otherwise we satisfy \ref{BanachO6}.
Otherwise we can give $x_1 = x$, $x_2 = f(x)$, $y_1 = y$, $y_2 = f(y)$ and since $x_1 \neq x_2$, $y_1 \neq y_2$ we satisfy \ref{eq:dMetric2} of \ref{BanachO5}.

  This implies that any output of $\clocal$ at the instance $(f', p, \eps', \lambda')$ can produce an output to the instance $(f, d, \eps, \lambda, c)$ of the
$\Banach$ problem. Therefore \\ $\Banach \in \CLS$.

  Now we are going to show the opposite direction and reduce $\clocal$ to $\Banach$. Starting from an instance  $(f, p, \eps, \lambda)$ of $\clocal$ we define for
any $x, y \in [0, 1]^3$,
\[ \kappa(x, y) = \min\left\{ - \frac{p(x)}{\eps}, - \frac{p(y)}{\eps} \right\} \]

\noindent We also remind the reader the definition of the \textit{discrete metric}
\[ d_S(x, y) = 1 \text{ if } x \neq y \text{ and } d_S(x, x) = 0 \]

\noindent Finally we define the \textit{smooth interpolation function} for $w \le 0$
\[ \Bsf{w} = \left( 1 - (\ceil{w} - w) \right) c^{\ceil{w}} + ( \left( \ceil{w} - w \right) ) c^{\ceil{w} + 1} \]

\noindent The basic observation about $\Bsf{\cdot}$ since $c < 1$ is that $c^{\ceil{\kappa(x, y)} + 1} \le B({\kappa(x, y)}) \le c^{\ceil{\kappa(x, y)}}$. \\

\noindent Based on these definitions we create the following instance of $\Banach$
\[ f' = f, d(x, y) = \Bsf{\kappa(x, y)} \cdot d_S(x, y), \eps' = \frac{1}{c}, \lambda' = \max \left\{ \lambda, \ceil{c^{-1 / \eps} \lambda \frac{\ln(1/c)}{\eps}} \right\},  c = 1 - 0.1 \eps \]
\noindent As in the previous reduction we have to show that any result of the $\Banach$ with input \\ $(f, d, \eps', \lambda, c)$ will give us a result of $\clocal$ with
input $(f, p, \eps, \lambda)$. \\

\noindent \ref{BanachO1} $\implies$ If $p(f(x)) \ge p(x)$ then $x$ satisfies \ref{cLocalO1}. Otherwise we can see that $\kappa(x, f(x)) = - p(x) / 2 \eps$ and $x \neq f(x)$ so
\begin{align*}
    d(x, f(x)) \le \eps' \Rightarrow \Bsf{\kappa(x, y)} \le \eps' \Rightarrow \left( \frac{p(x)}{\eps} \right) \log(1 / c) \le \log(\eps') \Rightarrow \frac{p(x)}{\eps} & \le \frac{\log(\eps')}{\log(1/c)} \\
  \Rightarrow p(x) & \le \eps
\end{align*}

\noindent so $p(f(x)) \ge 0 \ge p(x) - \eps$ and so $x$ satisfies \ref{cLocalO1}. 

\noindent \ref{BanachO3} $\implies$ As in the previous case we may assume that $p(f(x)) \le p(x) - \eps$ and that $p(f(y)) \le p(y) - \eps$. 
Without loss of generality we can assume that $p(x) > p(y)$. If also $p(f(x)) \ge p(f(y))$ then $\kappa(x, y) = - p(x) / \eps$ and
$\kappa(f(x), f(y)) = - p(f(x)) / \eps$. Therefore
\begin{align*}
    d(x, y) = \Bsf{\kappa(x, y)}, ~ d(f(x), f(y)) = \Bsf{\kappa(f(x), f(y))}
\end{align*}

  \noindent Now if \ref{BanachO3} is satisfied then
  \[ c^{-\floor{\frac{p(f(x))}{\eps}}} \ge d(f(x), f(y)) = \Bsf{\kappa(f(x), f(y))} > c \cdot \Bsf{\kappa(x, y)} = c \cdot d(x, y) \ge c^2 \cdot c^{-\floor{\frac{p(x)}{\eps}}} \]
  \[ \implies \floor{\frac{p(f(x))}{\eps}} \ge \floor{\frac{p(x)}{\eps}} - 2 \Rightarrow p(f(x)) \ge p(x) - 2 \eps ~~~~~ \footnote{At this point we should have set $\kappa(x, y) = \min\{- 2 p(x) / \eps, 2 p(y) / \eps\}$ to get the inequality $p(f(x)) \ge p(x) - \eps$ but this would complicate the calculations in the rest of the cases. It is clear though that we could scale every parameter so that $\eps$ becomes $2 \eps$ and nothing changes. } \]

\noindent Therefore $x$ satisfies \ref{cLocalO1}.

  Now similarly if $p(f(y)) > p(f(x))$ then $p(f(y)) > p(x) - \eps$. But by our assumption that $p(x) > p(y)$ we get $p(f(y)) > p(y) - \eps$. Therefore $y$
satisfies \ref{cLocalO1}.

\noindent \ref{BanachO4} $\implies$ \ref{cLocalO2}.

\noindent \ref{BanachO5} $\implies$ We will analyze the function $h(x) = c^{-x}$ when $x \in [0, 1/\eps]$. By the mean value theorem we have that the Lipschitz
constant $\ell_h$ of $h$ is less that $\max_{x \in [0, 1 / \eps]} h'(x)$. But
\[ h'(x) = \left( e^{-x \ln c} \right)' = \ln(1/c) c^{-x} \]

\noindent and because $c < 1$ we have that $\max_{x \in [0, 1/\eps]} h'(x) = c^{-1 / \eps} \ln(1/c)$.
  
Let now $\kappa(x_1, x_2) = - p(x_1)/\eps$ and $\kappa(y_1, y_2) = - p(y_1)/ \eps$. Since $x_1 \neq x_2$ and $y_1 \neq y_2$ we have $d(x_1, x_2) = \Bsf{c^{-p(x_1) / \eps}}$ and $d(y_1, y_2) = \Bsf{c^{-p(y_1) / \eps}}$. 
Since $\Bsf{\kappa(x, y)}$ is just an linear interpolation of points that belong to $c^{\kappa{x, y}}$ using the Mean Value Theorem we have that 
$\left| \Bsf{\kappa(x_1, x_2)} - \Bsf{\kappa(y_1, y_2)} \right| \le \max_{x \in [0, 1/\eps]} h'(x) \abs{\frac{p(x_1)}{\eps} - \frac{p(y_1)}{\eps}}$
\begin{align*}
  |d(x_1, x_2) - d(y_1, y_2)| = \left| \Bsf{\kappa(x_1, x_2)} - \Bsf{\kappa(y_1, y_2)} \right| & \le \left( \max_{x \in [0, 1/\eps]} h'(x) \right) \left| \frac{p(x_1)}{\eps} - \frac{p(y_1)}{\eps} \right| \\
  \Rightarrow |d(x_1, x_2) - d(y_1, y_2)| & \le c^{-1 / \eps} \frac{\ln(1/c)}{\eps} \left| p(x_1) - p(y_1) \right|
\end{align*}

  Now if $|p(x_1) - p(y_1)| > \lambda |x_1 - y_1|$ then $x_1, y_1$ satisfy \ref{cLocalO3} and we have a solution for $\clocal$. So $|p(x_1) - p(y_1)| \le \lambda |x_1 - y_1|$ and from
the last inequality we have that
\[ |d(x_1, x_2) - d(y_1, y_2)| \le c^{-1 / \eps} \lambda \frac{\ln(1/c)}{\eps} |x_1 - y_1| \]

But this contradicts with \ref{BanachO5} since $\lambda' = \max \left\{ \lambda, \ceil{c^{-1 / \eps} \lambda \frac{\ln(1/c)}{\eps}} \right\}$.

  Finally it is easy to see that the size of the arithmetic circuits that we used for this reduction is polynomial in the size of the input. The only function
that needs for explanation is that of $d$ and $\lambda'$. We start with the observation that both $c$, $c^{-1}$ are given and have descriptions of size only linear in the description of $\eps$, since $\eps$ is a rational constant. The
difficult term in the description of $d$ is the term $\Bsf{\kappa(x, y)}$. For this, we need to bound the size of $\kappa{x, y}$, let this bound be $A$. Then we can have precomputed the possible digits of $\ceil{\kappa(x, y)}$ using 
$\log(A)$ arithmetic circuits. Finally a final circuit combines the digits in order to get $\ceil{\kappa(x, y)}$. Now to compute $c^{\ceil{\kappa(x, y)}}$ for each $m_i 2^i$ of the $\log(A)$ digits of $\kappa(x, y)$ we compute the 
corresponding power $2^i$ with repeated squaring using $O(i)$ arithmetic gates. Then we combine the results such to compute $c^{\ceil{\kappa(x, y)}}$. This whole process needs $O(\log^2 A)$ arithmetic gates. Since $A \le 1/ \eps$ the
overall circuit for $d$ needs $\poly(1 / \eps)$ arithmetic gates. For $\lambda'$ we can also do a similar process but we have to bound $c^{-\kappa(x, y)}$. We can see that that 
$c^{-\kappa(x, y)} \le c^{-1/\eps} = (1 - 0.1 \eps)^{-1/\eps} \le e^{10}$. Therefore the size of $c^{-1/\eps}$ is bounded its ceil can be computed using a polynomial sized circuit.
\end{proof}

By inspecting the proof of the CLS-completeness of $\Banach$ we realize that, in the CLS-hardness part of the proof, we can actually guarantee that $d$ is a metric. We can thus also establish the CLS-completeness of $\Banachh$.

\begin{proof}[Proof of Theorem \ref{th:cBanach2}]
  Obviously because of Theorem \ref{th:cBanach1}, $\Banachh$ belongs to $\CLS$.

  For the opposite direction, we use the same reduction as in the proof of Theorem \ref{th:cBanach1}. We then prove that $d$ satisfies the desired properties. We
remind that we used the following instance of $\Banachh$ for the reduction
\[ f' = f, d(x, y) = \Bsf{\kappa(x, y)} \cdot d_S(x, y), \eps' = \frac{1}{\sqrt{c}}, \lambda' = \max \left\{ \lambda, \ceil{c^{-1 / \eps} \lambda \frac{\ln(1/c)}{\eps}} \right\},  c = 1 - 0.1 \eps \]

  We first prove that $d$ is always a distance metric.

\noindent \ref{eq:dMetric1} Obvious from the definition of $d$.

\noindent \ref{eq:dMetric2} If $x \neq y$ then $d_S(x, y) > 0$. Also always $c^{\kappa(x, y)} > 0$, therefore $d(x, y) > 0$. Now since $d_S(x, x) = 0$ we also have
$d(x, x) = 0$.

\noindent \ref{eq:dMetric3} It is obvious from the definition of $\kappa$ that $\kappa(x, y) = \kappa(y, x)$ and since $d_S$ is a distance metric, the same is true
for the $d_S$ and thus for $d$.

\noindent \ref{eq:dMetric4} Without loss of generality we assume that $p(x) \ge p(y)$. We consider the following cases
\begin{description}
  \item[$\mathbf{p(x) \ge p(z)}$] then we have $d(x, y) = d(x, z)$ and therefore obviously $d(x, y) \le d(x, z) + d(z, y)$.
  \item[$\mathbf{p(x) \le p(z)}$] then we have $d(x, y) = \Bsf{-\frac{p(x)}{2 \eps}}$, $d(x, z) = d(z, y) = \Bsf{-\frac{p(z)}{2 \eps}}$ but since $p(x) \le p(z)$ obviously
                                  $2 \Bsf{-\frac{p(z)}{2 \eps}} \ge \Bsf{-\frac{p(x)}{2 \eps}}$.
\end{description}

  Finally we will show the completeness of $([0, 1]^3, d)$. We first observe that for all $x \neq y$, $d(x, y) > 1$, this comes from the fact that $c < 1$ and so
$c^{-p(x)/\eps} > 1$.

  Now let $(x_n)$ be a Cauchy sequence then $\forall \delta > 0$, $\exists N \in \nats$ such that $\forall n, m > N$, $d(x_n, x_m) \le \delta$. We set
$\delta = 1/2$ then there exists $N \in \nats$ such that $\forall n, m > N$, $d(x_n, x_m) < 1/2$. But from the previous observation this implies $d(x_n, x_m) = 0$
and since $d$ defines a metric we get $x_n = x_m$. Therefore $(x_n)$ is constant for all $n > N$ and obviously converges. This means that every Cauchy sequence
converges and so $([0, 1]^3, d)$ is a complete metric space.
\end{proof}

\section{Proofs of Corollaries~\ref{cor:converse1} \ref{cor:converse2} }

\begin{prevproof}{Corollary}{cor:converse1}
  Let $d_{c, \eps/2}$ be the distance metric guaranteed by Theorem \ref{th:converse} with parameters $c$, $\eps/2$. Let also $(x_n)$ be the sequence produced by the \BIP. Since $f$ is a contraction with respect to $d_{c, \eps/2}$, we
have $d_{c, \eps/2}(x_n, x^*) \le \frac{c^n}{1 - c} d_{c, \eps/2}(x_0, x_1)$. If we make sure that $d_{c, \eps/2}(x_n, x_{n + 1}) \le \eps / 2$ then according to Theorem \ref{th:converse} 
$d(x_n, x^*) \le \eps$. So the number of steps that are needed are:
  \[ \frac{c^n}{1 - c} d_{c, \eps/2}(x_0, x_1) \le \frac{\eps}{2} \Leftrightarrow n \ge \frac{\log(d_{c, \eps/2}(x_0, x_1)) + \log((2 - 2 c)/\eps)}{\log(1/c)}. \]
\end{prevproof}

\begin{prevproof}{Corollary}{cor:converse2}
    Using Corollary \ref{cor:converse1}, we get that after $n = \frac{\log(d_{c,\delta/2}(x_0, f(x_0))) + \log((2 - 2 c)/\delta)}{\log(1/c)}$ iterations we will have
  $d(x_n, x^*) \le \delta$ or $d(x_{n+1}, x^*) \le \delta$.  Since in $\bar{B}(x^*, \delta)$, $f$ is a contraction with respect to $d$, it certainly must be that $d(x_{n+1}, x^*) \le \delta$. By the same token, $d(x_{n + 1 + m}, x^*) \le c^m d(x_{n+1}, x^*)$, for all $m > 0$. Therefore, to guarantee $d(x_{n + 1 + m}, x^*) \le \eps$, it suffices to take $m \ge \frac{-\log(1/\delta) + \log(1/\eps)}{\log(1/c)}$. So in total we need $n + 1 + m$ iterations, implying the number of iterations stated in the statement of the corollary.
\end{prevproof}

\end{document}